\documentclass[conference]{IEEEtran}
\usepackage{amsmath,amssymb,epsfig,psfrag,cite,subfigure}
\usepackage{graphicx}
\usepackage{epstopdf}

\usepackage{acronym}

\usepackage{bm}

\usepackage{pifont}

\usepackage{tikz}
\usetikzlibrary{calc, arrows.meta, positioning, decorations.pathmorphing}

\usepackage{tcolorbox}

\usepackage{accents}
\usepackage{hyperref}

\allowdisplaybreaks
\usepackage{soul}

\usepackage{amsthm}

\usepackage{nicefrac}

\usepackage{lipsum,multicol}

\usepackage{algorithm}
\usepackage{algpseudocode}
\algdef{SE}[SUBALG]{Indent}{EndIndent}{}{\algorithmicend\ }%
\algtext*{Indent}
\algtext*{EndIndent}

\usepackage{enumitem}
\usepackage{mwe}    
\usepackage{epsfig,psfrag}
\usepackage{subfigure}
\usepackage{color}
\usepackage{url}
\usepackage{mathtools,xparse}

\usepackage{gensymb}

\usepackage[multiple]{footmisc}

\usepackage{xr}
\makeatletter
\newcommand*{\addFileDependency}[1]{% argument=file name and extension
  \typeout{(#1)}
  \@addtofilelist{#1}
  \IfFileExists{#1}{}{\typeout{No file #1.}}
}
\makeatother

\newcommand*{\myexternaldocument}[1]{%
    \externaldocument{#1}%
    \addFileDependency{#1.tex}%
    \addFileDependency{#1.aux}%
}

\IfFileExists{supp.aux}{\myexternaldocument{supp}}{}

\usepackage{xcolor}

\newcommand\abss[1]{\lvert#1\rvert}

\newcommand{\Eee}{\mathbb{E}}

\newcommand{\snr}{{\rm{SNR}}}

\newcommand{\snra}{{\rm{SNR}}_{\rm{A}}}
\newcommand{\snre}{{\rm{SNR}}_{\rm{E}}}
\newcommand{\snrc}{{\rm{SNR}}_{\rm{C}}}

\newcommand{\rxii}{R_{\xi \xi}}
\newcommand{\rxiie}{R_{\xi \xi}^{\rm{E}}}
\newcommand{\deltat}{\Delta t}

\newcommand{\sigmaxi}{\sigma^2_{\xi}}

\newcommand{\Tsym}{ T_{\rm{sym}} }

\newcommand{\Tcp}{ T_{\rm{cp}} }

\newcommand{\FF}{ \mathbf{F} }

\newcommand{\deltaf}{ \Delta f }

\newcommand{\fc}{ f_c }

\newcommand{\stilde}{ \widetilde{s} }

\newcommand{\ytilde}{ \widetilde{y} }

\newcommand{\mtN}{{\mathcal{N}}}

\newcommand{\fdb}{ f_{\rm{3dB}}^{\rm{A} }}
\newcommand{\fdbe}{ f_{\rm{3dB}}^{\rm{E}} }
\newcommand{\fdbc}{ f_{\rm{3dB}}^{\rm{C}} }

\newcommand{\boldRe}{ \boldR^{\rm{E}}  }

\newcommand{\rect}[1]{ { \rm{rect} }\left(#1\right) }

\newcommand{\mtCN}{{\mathcal{CN}}}

\newcommand{\tracebigg}[1]{ {{{\rm{tr}}\big( #1 \big)}}} 

\newcommand{\realp}[1]{ \Re \left\{#1\right\}  }
\newcommand{\imp}[1]{ \Im \left\{#1\right\}  }

\newcommand{\Imatrix}{{ \boldsymbol{\mathrm{I}} }}

\newcommand{\bb}{ \mathbf{b} }

\newcommand{\boldzero}{{ {\boldsymbol{0}} }}
\newcommand{\boldone}{{ {\boldsymbol{1}} }}

\newcommand{\boldR}{ \mathbf{R} }
\newcommand{\boldRgen}{ \mathbf{R}^{\rm{gen}} }

\newcommand{\xigen}{ \xi^{\rm{gen}} }
\newcommand{\xiedyn}{ \widetilde{\xi}^{\rm{E}} }

\newcommand{\yy}{ \mathbf{y} }

\newcommand{\xx}{ \mathbf{x} }

\newcommand{\ww}{ \mathbf{w} }
\newcommand{\zz}{ \mathbf{z} }

\newcommand{\bxi}{ \boldsymbol{\xi} }

\newcommand{\bxia}{ \bxi^{\rm{A}} }

\newcommand{\trp}{\mathsf{T}}
\newcommand{\herm}{\mathsf{H}}

\newcommand{\cset}[2]{ \mathbb{C}^{#1 \times #2}  }

\newcommand{\rset}[2]{ \mathbb{R}^{#1 \times #2}  }

\newcommand{\conj}{ {\ast} }

\newcommand{\etab}{ {\boldsymbol{\eta}} }

\newcommand{\alphaa}{ \alpha^{\rm{A}} }
\newcommand{\taua}{ \tau^{\rm{A}} }

\newcommand{\alphasetilde}{ \widetilde{\alpha}^{\rm{E}} }

\newcommand{\alphase}{ \alpha^{\rm{E}} }
\newcommand{\tause}{ \tau^{\rm{E}} }

\newcommand{\ytildea}{ \ytilde^{\rm{A}}(t) }

\newcommand{\yya}{ \yy^{\rm{A}} }
\newcommand{\yye}{ \yy^{\rm{E}} }

\newcommand{\ycj}{ y^{\rm{C}} }
\newcommand{\zcj}{ z^{\rm{C}} }

\newcommand{\xiet}{ \xi^{\rm{E}} }

\newcommand{\xiestilde}{ \widetilde{\bxi}^{\rm{E}} }

\newcommand{\wwsetilde}{ \widetilde{\ww}^{\rm{E}} }

\newcommand{\wwse}{ \ww^{\rm{E}} }

\newcommand{\wwa}{ \ww^{\rm{A}} }
\newcommand{\zza}{ \zz^{\rm{A}} }

\newcommand{\zze}{ \zz^{\rm{E}} }

\newcommand{\phia}{ \phi^{\rm{A}} }
\newcommand{\phie}{ \phi^{\rm{E}} }
\newcommand{\phic}{ \phi^{\rm{C}} }

\makeatletter \renewcommand\d[1]{\ensuremath{%
		\;\mathrm{d}#1\@ifnextchar\d{\!}{}}}
\makeatother

\makeatletter
\newcommand*\rel@kern[1]{\kern#1\dimexpr\macc@kerna}
\newcommand*\widebar[1]{%
  \begingroup
  \def\mathaccent##1##2{%
    \rel@kern{0.8}%
    \overline{\rel@kern{-0.8}\macc@nucleus\rel@kern{0.2}}%
    \rel@kern{-0.2}%
  }%
  \macc@depth\@ne
  \let\math@bgroup\@empty \let\math@egroup\macc@set@skewchar
  \mathsurround\z@ \frozen@everymath{\mathgroup\macc@group\relax}%
  \macc@set@skewchar\relax
  \let\mathaccentV\macc@nested@a
  \macc@nested@a\relax111{#1}%
  \endgroup
}
\makeatother

\theoremstyle{remark}

\newtheoremstyle{mytheoremstyle} % name
    {\topsep}                    % Space above
    {\topsep}                    % Space below
    {\upshape}                   % Body font
    {.5em}                           % Indent amount
    {\itshape}                   % Theorem head font
    {.}                          % Punctuation after theorem head
    {.5em}                       % Space after theorem head
    {}  % Theorem head spec (can be left empty, meaning ‘normal’)

\theoremstyle{plain}

\newtheoremstyle{iremark}
  {\topsep}   % ABOVESPACE
  {\topsep}   % BELOWSPACE
  {\upshape}  % BODYFONT
  {0.2in}       % INDENT (empty value is the same as 0pt)
  {\itshape}  % HEADFONT
  {.}         % HEADPUNCT
  {5pt plus 1pt minus 1pt} % HEADSPACE
  {\thmname{#1}\thmnumber{ \itshape#2}\thmnote{ (#3)}} % CUSTOM-HEAD-SPEC

\newtheorem{proposition}{Proposition}
\theoremstyle{definition}

\newcommand{\etabe}{ {\boldsymbol{\eta}^{\rm{E}}} }   % Eve's parameter vector (wrapped so _0 subscripting is safe)
\newcommand{\xiat}{ {\xi^{\rm{A}}} }                  % Alice's differential PN scalar function

\acrodef{RIS}{reconfigurable intelligent surface}
\acrodef{SNR}{signal-to-noise ratio}
\acrodef{SINR}{signal-to-interference-plus-noise ratio}
\acrodef{ISAC}{integrated sensing and communication}
\acrodef{ISLAC}{integrated sensing, localization, and communication}
\acrodef{LOS}{line-of-sight}
\acrodef{NLOS}{non-line-of-sight}
\acrodef{AOA}{angle-of-arrival}
\acrodef{AOD}{angle-of-departure}

\acrodef{TOA}{time-of-arrival}
\acrodef{TDOA}{time-difference-of-arrival}

\acrodef{UE}{user equipment}
\acrodef{NF}{near-field}
\acrodef{BS}{base station}
\acrodef{AP}{access point}
\acrodef{MCRB}{misspecified Cram\'{e}r-Rao bound}
\acrodef{CRB}{Cram\'{e}r-Rao bound}
\acrodef{LB}{lower bound}
\acrodef{ML}{maximum-likelihood}
\acrodef{MML}{mismatched maximum-likelihood}
\acrodef{DL}{downlink}
\acrodef{UL}{uplink}
\acrodef{VRU}{vulnerable road user}
\acrodef{MIMO}{multiple-input multiple-output}
\acrodef{MISO}{multiple-input single-output}
\acrodef{SISO}{single-input single-output}
\acrodef{SIP}{shift invariance property}
\acrodef{FIM}{Fisher information matrix}
\acrodef{RMSE}{root mean-squared error}
\acrodef{AWGN}{additive white Gaussian noise}
\acrodef{ADMM}{alternating direction method of multipliers}
\acrodef{LS}{least-squares}
\acrodef{SOC}{second-order cone}
\acrodef{CFO}{carrier frequency offset}
\acrodef{i.i.d.}{independently and identically distributed}
\acrodef{MI}{mutual information}
\acrodef{SAC}[S\&C]{sensing and communication}
\acrodef{OTFS}{orthogonal time frequency space}

\acrodef{ULA}{uniform linear array}
\acrodef{FRO}{free-running oscillator}
\acrodef{PLL}{phase-locked loop}
\acrodef{PSLR}{peak-to-sidelobe level ratio}

\acrodef{PN}{phase noise}
\acrodef{LO}{local oscillator}
\acrodef{CP}{cyclic prefix}
\acrodef{OFDM}{orthogonal frequency-division multiplexing}

\acrodef{MF}{matched filtering}
\acrodef{RF}{reciprocal filtering}

\acrodef{CPE}{common phase error}
\acrodef{ICI}{intercarrier interference}
\acrodef{RCS}{radar cross section}

\acrodef{TX}{transmitter}
\acrodef{TCXO}{temperature-compensated crystal oscillator}
\acrodef{RX}{receiver}
\acrodef{QPSK}{quadrature phase-shift keying}
\acrodef{SPG}{sensing privacy gap}
\acrodef{SSB}{synchronization signal block}
\acrodef{PDF}{probability density function}
\acrodef{dB}{decibel}

\graphicspath{{./Figures/}}
\usepackage{pgfplots}
\usepgfplotslibrary{groupplots}
\pgfplotsset{compat=1.18}

\IEEEoverridecommandlockouts

\begin{document}
\bstctlcite{IEEEexample:BSTcontrol}

%%%%%%%%%%%%%%%%%% title page information %%%%%%%%%%%%%%%%%%
% \title{Role of Phase Noise in Sensing Privacy Protection in ISAC Systems}
\title{Exploiting Phase Noise for Sensing Privacy in ISAC Systems}

\author{\IEEEauthorblockN{Musa Furkan Keskin\IEEEauthorrefmark{1}, Kawon Han\IEEEauthorrefmark{2}, Henk Wymeersch\IEEEauthorrefmark{1}, Christos Masouros\IEEEauthorrefmark{3}
}
\vspace{0.1cm}
\IEEEauthorblockA{\IEEEauthorrefmark{1}Dept. of Electrical Engineering, Chalmers University of Technology, Sweden
} 
\IEEEauthorblockA{
\IEEEauthorrefmark{2}Ulsan National Institute of Science and Technology, Republic of Korea
}
\IEEEauthorblockA{
\IEEEauthorrefmark{3}Dept. of Electronic \& Electrical Engineering, University College London, UK
}\thanks{This work is supported by the SNS JU project 6G-DISAC under the EU's Horizon Europe research and innovation Program under Grant Agreement No 101139130 and by the Swedish Research Council (VR) through the project 6G-PERCEF under Grant 2024-04390.}
}

% %%%
% \makeatletter
% \def\@IEEEsectpunct{.\ \,} % (optional) fix section punctuation style
% \makeatother

% make the title area
\maketitle

%%%%%%%%%%%%%%%%%%%%%%%% abstract %%%%%%%%%%%%%%%%%%%%%%%%

\begin{abstract}
We investigate sensing privacy in \ac{OFDM} \ac{ISAC} systems under the impact of \ac{PN} arising from \ac{LO} imperfections. Specifically, we consider an ISAC scenario comprising a legitimate monostatic ISAC transceiver (Alice), an eavesdropper performing unauthorized bistatic sensing (Eve) and a communication user (UE), each equipped with a non-ideal \ac{LO}. To characterize sensing performance in the presence of \ac{PN}, we carry out a \ac{MCRB} analysis of monostatic and bistatic range estimation at Alice and Eve, whose differential \ac{PN} processes are self-correlated (delay-dependent) and cross-correlated (delay-independent) due to the use of a shared and an independent \ac{LO}, respectively. Simulation results reveal three-way trade-offs among legitimate monostatic sensing at Alice, unauthorized bistatic sensing at Eve and communication to the UE under \ac{PN}, governed by the \ac{LO} quality at Alice. Through the \ac{LO} asymmetry between Alice and Eve, worsening \ac{LO} quality at Alice can significantly enlarge sensing privacy gap in her favor, especially for nearby targets, with only a moderate reduction in data rate in noise-limited regimes.

\textit{Index Terms--} OFDM, ISAC, sensing privacy, bistatic sensing, phase noise, three-way trade-offs.
\end{abstract}

%%%%%%%%%%%%%%%%%%%%%%%%%%  body  %%%%%%%%%%%%%%%%%%%%%%%%%%
\section{Introduction}
Integrated sensing and communication (ISAC) has emerged as a cornerstone technology for future wireless networks, enabling simultaneous communication and environmental awareness through shared spectrum, waveforms and hardware resources \cite{Fan_ISAC_6G_JSAC_2022}. Such an integration can support numerous applications such as sensing-aided beam alignment \cite{ISAC_Veh_magazine_2025}, but it also raises new privacy concerns \cite{han2026next}. In particular, the same radiated ISAC waveform used for legitimate sensing can be exploited by unauthorized sensing receivers to passively infer location and velocity of targets from reflected signals \cite{6G_ISAC_Security_2025}. Hence, sensing privacy, i.e., protecting the privacy of sensed entities, constitutes a distinct design objective in ISAC systems rather than being merely an extension of communication security \cite{multidomain_magazine_2026}.

Recent works have addressed physical layer sensing privacy through ambiguity function (AF) shaping \cite{sensingSecure_ISAC_AF_2025,borui_secure_KLD_2025}, generative diffusion models \cite{genAI_secure_ISAC_2025} and exploitation of known-location scatterers \cite{chen2025sensingsecuritynearfieldisac}. A common take-away from these studies is that ISAC transmit signal  design in time, frequency or spatial domains can intentionally degrade unauthorized sensing while preserving the performance of legitimate sensing. In addition, the three-way trade-offs among communication rate, legitimate and unauthorized sensing must be rigorously optimized to avoid hampering communication functionality while damaging passive sensing eavesdroppers \cite{sensingSecure_ISAC_AF_2025}. Despite intense research efforts, a major gap persists in existing sensing privacy works with regard to the analysis of how hardware impairments (HWIs) can reshape sensing privacy and corresponding three-way trade-offs. This gap can be crucial for millimeter-wave and higher-frequency ISAC systems, where phase noise (PN), power amplifier nonlinearity and mutual coupling can severely impact both communication and sensing \cite{RF_JCS_2021}.

%%%%%%%%%%%%%%%%%%%%%%%%%%%%%%%%%%%%%%
\begin{figure}[t]
    \centering
    \includegraphics[width=\linewidth]{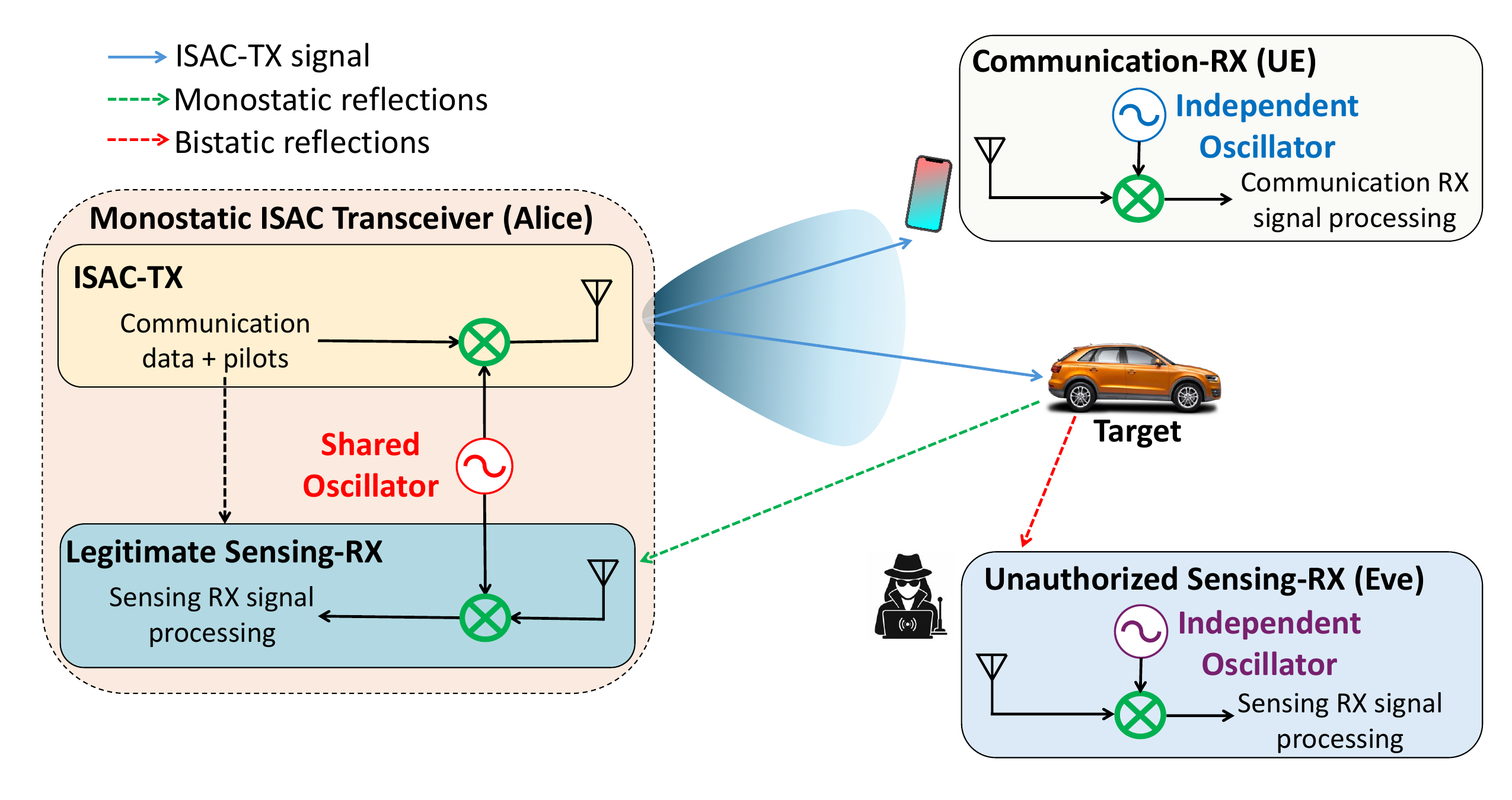}
        \vspace{-0.1in}
    \caption{Sensing-secure ISAC setting with legitimate monostatic sensing at Alice, unauthorized bistatic sensing at Eve and downlink communication to the UE. All nodes employ imperfect oscillators. The key architectural asymmetry is the shared \ac{LO} at Alice versus the independent \acp{LO} at Eve and the UE.}
    \label{fig_scenario}
\end{figure}
%%%%%%%%%%%%%%%%%%%%%%%%%%%%%%%%%%%%%%

Among HWIs, PN poses a particularly formidable challenge as it has a rapidly time-varying nature and its severity increases with carrier frequency \cite{PN_CohBw_2021_TWC}. In monostatic sensing, the transmitter and sensing receiver share the same \ac{LO}, which implies that downconversion produces a self-referenced differential \ac{PN} (DPN) process whose statistics depend on the target delay, i.e., the range correlation effect \cite{Range_Correlation_93,OFDM_PN_TSP_Exploitation_2023,PN_OFDM_ISAC_RadarConf_2023}. In bistatic sensing, by contrast, the observed DPN becomes the difference of transmit and receive \ac{PN} from separate \acp{LO} due to an independent receive \ac{LO}. Hence, unlike monostatic sensing where closer targets yield smaller DPN fluctuations through range correlation, bistatic DPN statistics are not range-dependent. In \cite{Kawon_Sync_PN_DISAC_RadarConf25}, the authors have demonstrated the significantly higher vulnerability of bistatic sensing to PN than monostatic sensing, which results from the above-mentioned \ac{LO} asymmetry. However, it remains unexplored how this asymmetry affects sensing privacy gap between legitimate monostatic sensing and unauthorized bistatic sensing in ISAC systems.

Motivated by this asymmetry, we investigate a sensing-secure ISAC system comprising a legitimate monostatic ISAC transceiver (Alice), an unauthorized passive bistatic sensing receiver (Eve) and a communication user equipment (UE), each equipped with a non-ideal \ac{LO} impaired by \ac{PN}. Unlike
waveform-domain sensing privacy-enhancing methods, our focus is on whether \textit{hardware-induced LO asymmetry itself can act as a
privacy-enhancing mechanism.} The main contributions are: \textit{(i)} We propose a novel three-party OFDM ISAC framework that explicitly captures the shared-LO architecture at Alice and
independent-LO architectures at Eve and the UE; \textit{(ii)} We provide a rigorous characterization of the differential PN statistics at Eve, and prove their
delay-independence, in contrast to the delay-dependent PN statistics at Alice; \textit{(iii)} We reveal the three-way trade-offs among communication performance, legitimate sensing and unauthorized sensing under \ac{PN}, and show that \ac{LO} asymmetry can be leveraged to enlarge sensing privacy gap in favor of Alice while incurring a moderate communication rate penalty in noise-limited regimes.

%%%%%%%%%%%%%%%%%%%%%%%%%%%%

%%%%%%%%%%%%%%%%%%%%%%%%%%%%

%%%%%%%%%%%%%%%%%%%%%%%%%%
\section{System Model}\label{sec_system_model}
\subsection{Scenario Description}\label{sec_scenario}
We consider an ISAC system consisting of both legitimate and malicious entities, all equipped with a single antenna, as shown in Fig.~\ref{fig_scenario}. The legitimate entities include \textit{(i)} the monostatic ISAC transceiver (Alice) equipped with an ISAC \ac{TX} and a sensing \ac{RX} on the same hardware platform, and \textit{(ii)} the communication \ac{RX} (\ac{UE}). Alice transmits an \ac{OFDM} signal to the \ac{UE} and simultaneously uses the backscattered echoes for monostatic sensing. Meanwhile, an unauthorized RX (Eve) acts as a passive bistatic sensing receiver eavesdropping on the legitimate ISAC system by exploiting the same illumination to infer target parameters independently\footnote{With single-antenna nodes, the sensing task involves only delay/range estimation. Angle estimation and localization require multiple antennas or additional temporal/spatial diversity and are left for future extensions.}. Moreover, the ISAC TX and the sensing RX at Alice share the same \ac{LO} \cite{OFDM_PN_TSP_Exploitation_2023,PN_OFDM_ISAC_RadarConf_2023}, whereas the UE and Eve each employ independent \acp{LO} at their remote receivers \cite{Kawon_Sync_PN_DISAC_RadarConf25}. Due to hardware imperfections, all the \acp{LO} are assumed to be non-ideal and impaired by \ac{PN} \cite{Demir_PN_2000,PN_OFDM_PLL_TCOM_2007}. The \acp{LO} at Alice, Eve and the UE have the \ac{PN} processes $\phia(t)$, $\phie(t)$ and $\phic(t)$, respectively.

\subsection{Transmit Signal Model}\label{sec_transmit}
The \ac{TX} sends one \ac{OFDM} symbol with $N$ subcarriers, subcarrier spacing $\deltaf=1/T$, useful duration $T$, \ac{CP} duration $\Tcp$ and total duration $\Tsym=\Tcp+T$ \cite{RadCom_Proc_IEEE_2011,sensingSecure_ISAC_AF_2025}. The occupied bandwidth is $B=N\deltaf$. The complex baseband signal is 
\begin{align}\label{eq_ofdm_baseband}
s(t)=\frac{1}{\sqrt{N}}\sum_{n=0}^{N-1}x_n e^{j2\pi n\deltaf t}\rect{\frac{t}{\Tsym}},
\end{align}
where $x_n$ denotes the transmitted symbol on subcarrier $n$ and $\rect{t}=1$ for $t\in[0,1]$ and zero otherwise. With \ac{PN}, the passband transmit signal is \cite{PN_2006,OFDM_PN_TSP_Exploitation_2023}
\begin{align}\label{eq_passband_st}
\stilde(t)=\Re\{s(t)e^{j(2\pi\fc t+\phia(t))}\}.
\end{align}

\subsection{Received Signal Model at Alice}\label{sec_radar_rec_alice}
For a point target with monostatic round-trip delay $\taua$ and complex gain $\alphaa$, the passband echo at Alice is \cite{OFDM_PN_TSP_Exploitation_2023}
\begin{align}\label{eq_rec_passband}
    \ytildea=\Re\{\alphaa s(t-\taua)e^{j[2\pi\fc(t-\taua)+\phia(t-\taua)]}\}.
\end{align}
Downconversion by Alice's shared noisy \ac{LO} gives \cite{PN_SI_TSP_2017,OFDM_PN_TSP_Exploitation_2023}
\begin{align}\label{eq_rec_baseband}
y^{\rm A}(t)=\alphaa s(t-\taua)e^{-j2\pi\fc\taua}w^{\rm A}(t,\taua) \,,
\end{align}
where $w^{\rm A}(t,\tau)= e^{-j\xiat(t,\tau)}$ with
\begin{align}\label{eq_pn_def_mono}
\xiat(t, \tau) \triangleq \phia(t)-\phia(t-\tau) \,.
\end{align}
Sampling at $t_\ell=\Tcp+\ell T/N$, $\ell=0,\ldots,N-1$, and absorbing deterministic carrier phase into $\alphaa$, we obtain the time-domain observation vector \cite{MIMO_OFDM_ICI_JSTSP_2021,OFDM_PN_TSP_Exploitation_2023}
\begin{align}\label{eq_y_all}
    \yya=\alphaa\,\wwa\odot \FF_N^\herm\big(\xx\odot\bb(\taua)\big)+\zza \,,
\end{align}
where $\yya\in\cset{N}{1}$, $\xx=[x_0,\ldots,x_{N-1}]^\trp$, $\FF_N$ is the unitary DFT matrix, $\zza\sim\mtCN(\boldzero,\sigma^2\Imatrix)$,
\begin{align}\label{eq_steer_delay}
\bb(\tau)=\big[1,e^{-j2\pi\deltaf\tau},\ldots,e^{-j2\pi(N-1)\deltaf\tau}\big]^\trp,
\end{align}
and $[\wwa]_\ell=e^{-j \xiat(t_\ell,\taua)}$. The standard \ac{CP} condition $\Tcp\geq\taua$ is assumed \cite{OFDM_Radar_Phd_2014}.

\subsection{Received Signal Model at Eve}\label{sec_radar_rec_eve}
Unlike the shared-oscillator transceiver architecture of Alice, Eve employs an independent \ac{LO}. 
%and no usable \ac{LOS} reference path from the \ac{TX}; only the target-scattered echo reaches her sensing front end. 
To isolate the impact of \ac{PN} on sensing privacy and disentangle it from confounding factors such as timing or data symbol uncertainty, we model Eve as the worst-case (most powerful) adversary by adopting two assumptions: \textit{(i)} Eve has acquired and maintains \ac{OFDM} frame timing using the same standard-compliant over-the-air synchronization procedures relied upon by the \ac{UE} \cite{bistatic_OTA_2025,passive_radar_tutorial_2019}, so Eve is perfectly time synchronized with Alice and her sampling grid coincides with that of Alice, i.e., $t_\ell=\Tcp+\ell T/N$; and \textit{(ii)} Eve has perfect knowledge of the transmit symbols $\xx$. Under \textit{(i)} and \textit{(ii)}, the only remaining structural difference between the observations at Alice and Eve results from the independent \ac{PN} process acting on Eve's samples due to her independent \ac{LO}; any sensing privacy gap  therefore originates from this \ac{LO} asymmetry alone.

Under these assumptions, the received observation vector at Eve after \ac{CP} removal is \cite{Kawon_Sync_PN_DISAC_RadarConf25}
\begin{align}\label{eq_rec_baseband_eve_y}
    \yye=\alphase\,\wwse\odot \FF_N^\herm\big(\xx\odot\bb(\tause)\big)+\zze \,,
\end{align}
where $\yye\in\cset{N}{1}$, $\tause$ is the bistatic delay, i.e., the delay of the two-segment \ac{TX}-target-Eve path, $\zze\sim\mtCN(\boldzero,\sigma^2\Imatrix)$ and
\begin{align}\label{eq_pn_def_bistatic}
    [\wwse]_\ell&=e^{-j\xiet(t_\ell,\tause)} \,, \\ \label{eq_dpn_eve}
    \xiet(t,\tau)&  \triangleq \phie(t)-\phia(t-\tau) \,.
\end{align}
The \ac{CP} condition $\Tcp\geq\tause$ keeps $t_\ell-\tause\geq0$. The received signal at Eve in \eqref{eq_rec_baseband_eve_y} differs from that at Alice in \eqref{eq_y_all} in that the DPN in \eqref{eq_dpn_eve} contains two independent \ac{LO} processes rather than a delayed copy of one shared \ac{LO} as in \eqref{eq_pn_def_mono}. Secs.~\ref{sec_pn_stat} and~\ref{sec_perf_met} will show how this asymmetry propagates into the PN statistics and the sensing bounds.

\subsection{Received Signal Model at the UE}\label{sec_radar_rec_ue}
For the UE, we employ the standard model involving \ac{CPE} and \ac{ICI} to obtain the observation on subcarrier $n$ \cite{PN_OFDM_PLL_TCOM_2007,PN_OFDM_rate_2011}
\begin{align}\label{eq_comm_final}
    \ycj_n=x_n h_nJ_0+\sum_{k=0, k \neq n}^{N-1}x_kh_kJ_{n-k}+\zcj_n \,,
\end{align}
where $h_n$ denotes the communication channel on subcarrier $n$, $\zcj_n\sim\mtCN(0,\sigma^2)$ and \cite{PN_OFDM_PLL_TCOM_2007,PN_OFDM_rate_2011}
\begin{align}\label{eq_pn_dft}
J_i&=\frac{1}{N}\sum_{\ell=0}^{N-1}e^{j\theta_\ell}e^{-j2\pi \ell i/N} \,,
\\ \label{eq_theta_ell}
\theta_\ell&=[\phia(t)-\phic(t)]_{t=t_\ell} \,.
\end{align}
Here, $J_0$ is the \ac{CPE} coefficient, $\{J_i\}_{i\neq0}$ generate \ac{ICI} and $\{\theta_\ell\}$ represent the PN samples containing contributions from the \ac{PN} processes of both Alice and the UE, where $t_\ell=\Tcp+\ell T/N$.\footnote{In \eqref{eq_comm_final}--\eqref{eq_theta_ell}, we adopt the standard narrowband \ac{PN} approximation $\phia(t-\tau_h)\approx\phia(t)$, where $\tau_h$ is the delay spread of the communication channel. This approximation is accurate whenever the $3 \, \rm{dB}$ bandwidth of the Lorentzian oscillator spectrum is small compared to the subcarrier spacing, i.e., $\fdb \ll \deltaf$, or, equivalently, $\tau_h \ll 1/\fdb$, which holds for most practical scenarios \cite{PN_OFDM_PLL_TCOM_2007,PN_OFDM_rate_2011}. Moreover, unlike \eqref{eq_rec_baseband_eve_y}--\eqref{eq_dpn_eve}, we deliberately omit the absolute path delays from the argument of $\phia(t)$ in \eqref{eq_theta_ell} since the absolute delays do not affect the resulting PN statistics, as will be shown in Sec.~\ref{sec_eve_pn_stat}. Hence, the absolute sampling instants $t_\ell$ are immaterial in this model. What matters instead is that $\theta_\ell$ is formed by the combination of the realization of two independent \ac{PN} processes, each consisting of $N$ samples spaced by $T/N$, which completely determines the relevant statistical behavior \cite{PN_OFDM_PLL_TCOM_2007}.}

\subsection{Problem Statement}\label{sec_problem}
Given \eqref{eq_y_all}, \eqref{eq_rec_baseband_eve_y} and \eqref{eq_comm_final}, our goal is to evaluate the impact of \ac{PN} on legitimate sensing, unauthorized sensing and downlink communication under a wide range of conditions including SNR, oscillator quality and target range. To assess sensing performance at Alice and Eve, we employ both the standard \ac{CRB} (for PN-free baselines) and the \ac{MCRB} \cite{Fortunati2017} (for PN-ignorant processing) while communication performance is evaluated through the \ac{PN}-contaminated \ac{SINR} expression \cite{PN_OFDM_rate_2011}\footnote{With perfect knowledge of the channel $\{h_n\}$, the UE performs demodulation of the transmit data symbols $\xx$ under the impact of $\ac{PN}$ in \eqref{eq_theta_ell} (i.e., without carrying out CPE correction in alignment with PN-ignorant processing at Alice and Eve).}.

%%%%%%%%%%%%%%%%%%%%%%%%%%%%%%%%%%%%%
\section{Phase Noise Statistics}\label{sec_pn_stat}
This section provides a statistical characterization of the \ac{PN} processes at Alice and Eve, which will be used in Sec.~\ref{sec_perf_met} to derive performance metrics for sensing at the respective receivers. For the \ac{FRO} models adopted in this work, the Wiener \ac{PN} processes are modeled as \cite[Corr.~7.1]{Demir_PN_2000}, \cite[Sec.~III-A]{PN_OFDM_PLL_TCOM_2007}, \cite[Sec.~III-A]{PN_OFDM_rate_2011} 
\begin{align}
\phia(t) &\sim \mtN(0, 4 \pi \fdb t) \,, \\
\phie(t) &\sim \mtN(0, 4 \pi \fdbe t) \,, \\ 
\phic(t) &\sim \mtN(0, 4 \pi \fdbc t) \,,
\end{align}
where $\fdb$, $\fdbe$ and $\fdbc$ denote the $3 \, \rm{dB}$ bandwidth of the Lorentzian oscillator spectrum at the respective \acp{LO}.

\subsection{Alice: Shared-LO and Self-Referenced PN}
In \eqref{eq_pn_def_mono} and \eqref{eq_y_all}, let us define $\bxia \in \rset{N}{1}$ such that $\wwa = e^{-j \bxia}$. The statistics of the \textit{self-referenced} \ac{PN} $\bxia$ are given by \cite{OFDM_PN_TSP_Exploitation_2023,PN_OFDM_ISAC_RadarConf_2023}
\begin{align}\label{eq_bxi_stat}
    \bxia \sim \mtN(\boldzero, \boldR(\taua)) \,,
\end{align}
where $\boldR(\taua) \in \rset{N}{N}$ is the \textit{delay-dependent} covariance matrix of $\bxia$, with the entries
\begin{align}\label{eq_rtau_entries}
    \left[ \boldR(\tau) \right]_{n_1,n_2} = \rxii\big((n_1-n_2) T/N, \tau \big) \,.
\end{align}
In \eqref{eq_rtau_entries}, we have \cite{OFDM_PN_TSP_Exploitation_2023}
\begin{align}\label{eq_exp_xi2}
     \rxii(\deltat, \tau) &= \frac{ \sigmaxi(\tau+\deltat) + \sigmaxi(\tau-\deltat) }{2} - \sigmaxi(\deltat) \,,
\end{align}
where $\sigmaxi(\tau) = 4 \pi \fdb \abss{\tau}$. This yields 
\begin{align}\label{eq_alice_triangular}
    \rxii(\deltat,\tau)=4\pi\fdb\max(\tau-|\deltat|,0) \,.
\end{align}
The range correlation effect manifests in \eqref{eq_alice_triangular}, i.e., smaller target delays make the shared-\ac{LO} samples more correlated and reduce the variance of the DPN process at Alice \cite{OFDM_PN_TSP_Exploitation_2023}.

\subsection{Eve: Independent-LO and Cross-Referenced PN}\label{sec_eve_pn_stat}
\subsubsection{Reparameterization of \eqref{eq_rec_baseband_eve_y}}
To derive the statistics of $\xiet(t,\tau)$ in \eqref{eq_dpn_eve},  we first rewrite it for $t=\Tcp+u$ and $u \geq 0$ as
\begin{align}\label{eq_eve_split}
\xiet(\Tcp+u,\tau)
&=\underbrace{\phie(\Tcp)-\phia(\Tcp-\tau)}_{\rm{common~phase}} + \xiedyn(u) \,,
\end{align}
where
\begin{align} \label{eq_psiu}
    \xiedyn(u) &\triangleq [\phie(\Tcp+u)-\phie(\Tcp)]
    \\ \nonumber &~~~~-[\phia(\Tcp-\tau+u)-\phia(\Tcp-\tau)] \,.
\end{align}
The common phase term in \eqref{eq_eve_split} is independent of the time index $u$, rotates all $N$ samples in \eqref{eq_rec_baseband_eve_y} by the same scalar and thus can be absorbed into the channel gain $\alphase$. Hence, we can reparameterize \eqref{eq_rec_baseband_eve_y} into
\begin{align}\label{eq_rec_baseband_eve_y_re}
    \yye=\alphasetilde\,\wwsetilde\odot \FF_N^\herm\big(\xx\odot\bb(\tause)\big)+\zze \,,
\end{align}
where
\begin{align}
\alphasetilde &\triangleq \alphase e^{-j[\phie(\Tcp)-\phia(\Tcp-\tau)]} \,, 
\\ \label{eq_wwsetilde}
    [\wwsetilde]_\ell &\triangleq  e^{-j\xiedyn(u_\ell)} \,,~ u_\ell = \ell T/N \,. 
    \end{align}

\subsubsection{Delay-Independence of PN Statistics}
We define $\xiestilde \in \rset{N}{1}$ such that $\wwsetilde=e^{-j\xiestilde}$ in \eqref{eq_wwsetilde} and provide the following result to characterize its statistics.

\begin{proposition}\label{lemma_pn_stat_general}
    The correlation function of $\xiedyn(t)$ in \eqref{eq_psiu} is given by
    \begin{align} \nonumber
        &\rxiie(t_1,t_2)
        \triangleq \Eee\big[\xiedyn(t_1)\,\xiedyn(t_2)\big] = 4\pi(\fdb+\fdbe)\min(t_1,t_2).
    \end{align}
\end{proposition}
\begin{proof}
    It follows similar steps to those in \cite[Sec.~S-I]{OFDM_PN_TSP_Exploitation_2023}.
\end{proof}
From Proposition~\ref{lemma_pn_stat_general}, we note that as opposed to the delay-dependent nature of the correlation function \eqref{eq_alice_triangular} of the DPN process $\xiat(t, \tau)$ in \eqref{eq_pn_def_mono} at Alice, the DPN process $\xiedyn(t)$ in \eqref{eq_psiu} at Eve has delay-independent statistics. Based on Proposition~\ref{lemma_pn_stat_general} and \eqref{eq_wwsetilde}, the PN statistics at Eve can be obtained as
\begin{align}\label{eq_pn_stat_eve_dynamic}
    \xiestilde\sim\mtN(\boldzero,\boldRe) \,,
\end{align}
where $\boldRe \in \rset{N}{N}$ is a delay-independent covariance with
\begin{align} \label{eq_dpn_eve_delay_ind}  [\boldRe]_{n_1,n_2}=4\pi(\fdb+\fdbe)\min(u_{n_1},u_{n_2})\,.
\end{align}
Unlike the covariance of the \textit{self-referenced} \ac{PN} $\bxia$ in \eqref{eq_rtau_entries}, which depends on the time lag between the sampling instants (but not on those instants themselves, thus exhibiting Toeplitz property \cite{OFDM_PN_TSP_Exploitation_2023}), the covariance of the \textit{cross-referenced} \ac{PN} $\xiestilde$ in \eqref{eq_dpn_eve_delay_ind} depends on the sampling instants themselves, due to independent \ac{PN} processes $\phia(t)$ and $\phie(t)$ entering $\xiedyn(t)$ in \eqref{eq_psiu}. As a result, the variance of the DPN process at Eve scales with the total bandwidth of the two \acp{LO} \cite{PN_OFDM_PLL_TCOM_2007} and is independent of the target delay, as established in Proposition~\ref{lemma_pn_stat_general}.

\subsubsection{Contrast Between Alice and Eve}
Given the aforementioned analysis, the sensing privacy related asymmetry results from the fact that the PN at Eve, represented by $\xiedyn(u)$ in \eqref{eq_psiu}, includes the sum of two independent \ac{PN} walks over the OFDM symbol duration, whereas Alice observes a self-referenced increment over the target delay in \eqref{eq_pn_def_mono}. For nearby targets, this distinction can be substantial as Alice compares $\phia(t)$ and $\phia(t-\taua)$ over a small interval while Eve compares the random walks of $\phie(t)$ and $\phia(t)$ through two unrelated \acp{LO} whose relative phase continues to drift across the whole OFDM symbol.

To illustrate this contrast, we plot in Fig.~\ref{fig_pn_stats_contrast} the mean accumulated PN variance of Alice and Eve over $N$ fast-time samples, defined as 
\begin{align} \nonumber
    \frac{1}{N}\tracebigg{\boldRgen(\tau)} &= \Eee \bigg\{ \sum_{\ell = 0}^{N-1} \abss{\xigen(t_\ell, \tau)}^2 /N \bigg\} \,,
    \\ \label{eq_acc_PN}
    &= \begin{cases}
        4 \pi \fdb \tau \,, &\rm{Alice} \\
        2 \pi (\fdb + \fdbe) T \,, &\rm{Eve}
    \end{cases}
\end{align}
 for $t_\ell = \ell T/N$, where $\boldRgen(\tau) \in \{ \boldR(\tau), \boldRe \}$ and $\xigen(t,\tau) \in \{ \xiat(t,\tau), \xiedyn(t)  \}$, corresponding to Alice and Eve, respectively. As seen from the figure and \eqref{eq_acc_PN}, the PN variance of Eve is independent of target delay $\tau$ and scales with $T$ as $\fdb$ increases while the PN variance of Alice depends on $\tau$ and scales with $\tau$. Since $\tau \leq \Tcp \ll T$, Eve experiences a much faster growth in PN distortion than Alice as $\fdb$ increases, which creates the key privacy mechanism. Namely, increasing PN severity can substantially degrade the sensing performance of Eve while leaving Alice's comparatively less affected, thereby widening the sensing privacy gap.

%%%%%%%%%%%%%%%%%%%%%%%%%%%%%%%%%%%%%%%%%%
% target RCS
\begin{figure}[t]
	\centering
    %\vspace{-0.1in}
	\includegraphics[width=1\linewidth]{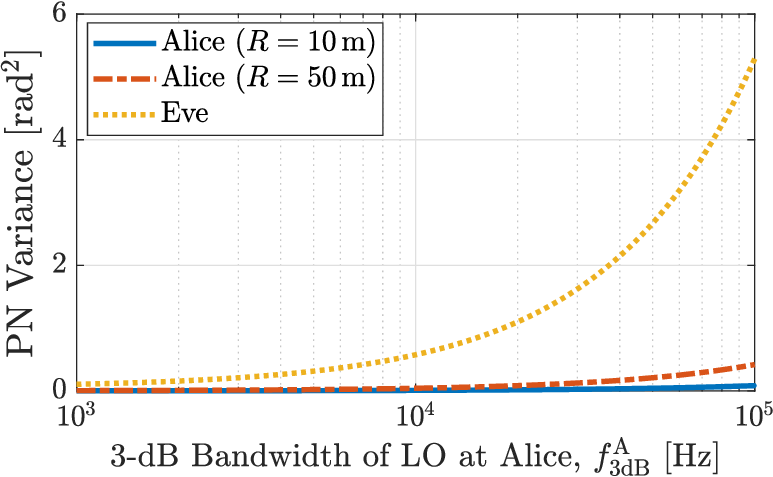}
	\vspace{-0.2in}
	\caption{Mean accumulated PN variance in \eqref{eq_acc_PN}, belonging to Eve and Alice at different ranges $R$, with respect to $\fdb$, where $\fdbe = 1 \, \rm{kHz}$.} 
	\label{fig_pn_stats_contrast}
	%\vspace{-0.1in}
\end{figure}
%%%%%%%%%%%%%%%%%%%%%%%%%%%%%%%%%%%%%%%%%%

%%%%%%%%%%%%%%%%%%%%%%%%%%%%%%%%%%%%%
\section{Performance Metrics}\label{sec_perf_met}
Based on the \ac{PN} statistics in Sec.~\ref{sec_pn_stat}, in this section we derive the sensing and communication performance metrics using the observations in \eqref{eq_y_all}, \eqref{eq_rec_baseband_eve_y_re} and \eqref{eq_comm_final}, following the standard \ac{CRB} and \ac{MCRB} machinery of \cite{Fortunati2017,OFDM_PN_TSP_Exploitation_2023,PN_OFDM_ISAC_RadarConf_2023}.

\subsection{Sensing Performance of Alice}\label{sec_sens_alice}
Based on the signal model in \eqref{eq_y_all}, we derive the \ac{CRB} and \ac{MCRB} on estimation of $\taua$ \cite[Sec.~III]{PN_OFDM_ISAC_RadarConf_2023}.

\subsubsection{CRB without Phase Noise}\label{sec_crb_alice}
This corresponds to the \textit{PN-free} case. The signal model becomes
\begin{align}\label{eq_sig_det_crb_alice}
    \yya = \alphaa \, \FF_N^\herm \big(\xx \odot \bb(\taua)  \big) + \zza \,,
\end{align}
where the unknown parameter vector is given by
\begin{align}\label{eq_eta_crb_alice}
    \etab = \left[ \taua, \realp{\alphaa}, \imp{\alphaa} \right]^\trp \in \rset{3}{1} \,.
\end{align}
Then, the derivation of the CRB for $\taua$ follows the same steps as those presented in \cite[Sec.~III-A]{PN_OFDM_ISAC_RadarConf_2023}.

\subsubsection{Misspecified CRB under Phase Noise}\label{sec_mcrb_alice}
This corresponds to the \textit{PN-ignorant} case. The true model with \ac{PN} is \eqref{eq_y_all} for a given \ac{PN} realization $\bxia$, while the assumed model without \ac{PN} corresponds to \eqref{eq_sig_det_crb_alice}, with the misspecified parametric \ac{PDF} parameters given by \eqref{eq_eta_crb_alice}. The remaining derivations follow those in \cite[Sec.~III-C]{PN_OFDM_ISAC_RadarConf_2023}.

\subsection{Sensing Performance of Eve}\label{sec_sens_eve}
To derive the bounds on estimation of $\tause$, we use the signal model in \eqref{eq_rec_baseband_eve_y_re} and the PN statistics in \eqref{eq_pn_stat_eve_dynamic}.
%where $\wwsetilde=e^{-j\xiestilde}$.

\subsubsection{CRB without Phase Noise}\label{sec_crb_eve}
The \textit{PN-free} case ($\wwsetilde=\boldone$) reduces \eqref{eq_rec_baseband_eve_y_re} to
\begin{align}\label{eq_sig_det_crb}
    \yye = \alphasetilde\,\FF_N^\herm \big(\xx\odot\bb(\tause)\big) + \zze \,,
\end{align}
where the unknown parameter vector is given by
\begin{align}\label{eq_eta_eve}
    \etabe =   \big[\tause, \realp{\alphasetilde}, \imp{\alphasetilde}\big]^\trp \in \rset{3}{1} \,.
\end{align}
%Here, the data-symbol phases $\thetadb\triangleq\{{\theta}_{{\rm{d}}_n}\}_{n\in\setNd}\in[-\pi,\pi)^{\abss{\setNd}}$ are continuous deterministic nuisance parameters.

\subsubsection{Misspecified CRB under Phase Noise}\label{sec_mcrb_eve}
In this \textit{PN-ignorant} case, the true model with \ac{PN} is \eqref{eq_rec_baseband_eve_y_re} for a given \ac{PN} realization $\xiestilde$ with the statistics in \eqref{eq_pn_stat_eve_dynamic}, while the assumed model without \ac{PN} is given by \eqref{eq_sig_det_crb}, with the parameters of the misspecified parametric \ac{PDF} defined in \eqref{eq_eta_eve}.

%%%%%%%%%%%%%%%%%%%%%%%%%%%%%%%%%%%%%
\subsection{Communication Performance with the UE}
To quantify the communication rate under PN, we first evaluate the \ac{SINR} per subcarrier. To that end, we assume that the data symbols satisfy $\Eee\{x_i\} = 0, \, \forall i$ and $\Eee\{x_i x_j^\conj \} = \delta(i-j) \, \forall i,j$. Then, given \eqref{eq_comm_final}, the SINR at subcarrier $n$ conditioned on fixed channels $h_n$  can be computed as \cite{PN_OFDM_rate_2011}
\begin{align}\label{eq_sinr}
    \gamma_n = \frac{ \abss{h_n}^2 \Eee\{\abss{J_0}^2\} }{ \sum_{k=0, k \neq n}^{N-1}  \abss{h_k}^2 \Eee\{\abss{J_{n-k}}^2\} + \sigma^2 } \,,
\end{align}
where, from \eqref{eq_pn_dft}, we have
\begin{align} \nonumber
    \Eee\{\abss{J_i}^2\} = \frac{1}{N^2} \sum_{\ell_1=0}^{N-1} \sum_{\ell_2=0}^{N-1} \Eee\big\{ e^{j (\theta_{\ell_1} - \theta_{\ell_2})} \big\} e^{-j 2 \pi (\ell_1-\ell_2)i/N} \,.
\end{align}
Using \eqref{eq_theta_ell}, the expectation evaluates to \cite[Eq.~(10)]{PN_2006}, \cite[Eq.~(19)]{PN_OFDM_PLL_TCOM_2007}
\begin{align}\label{eq_theta_fro_comm}
    \Eee\big\{ e^{j (\theta_{\ell_1} - \theta_{\ell_2})} \big\} = e^{ - 2 \pi (\fdb + \fdbc) \abss{\ell_1-\ell_2} T/N } \,.
\end{align}
As in Eve's PN statistics in \eqref{eq_dpn_eve_delay_ind}, the PN statistics at the UE capture the impact of two independent \acp{LO} whose total bandwidth determines the severity of \ac{PN} in \eqref{eq_theta_fro_comm}. With the SINRs in \eqref{eq_sinr}, the average achievable rate is \cite{PN_OFDM_rate_2011}
\begin{align}\label{eq_cap}
    C = \frac{1}{N} \sum_{n = 0}^{N-1}  \log_2( 1 + \gamma_n) \,.
\end{align}

%%%%%%%%%%%%%%%%%%%%%%%%%%%%%%%%%%%%%
\section{Simulation Results}\label{sec_results}
For performance evaluation, we consider the default setup in Table~\ref{tab_parameters} in accordance with 5G NR FR2 parameters \cite{TR_38211} and typical oscillator parameters \cite{PN_OFDM_PLL_TCOM_2007,PN_OFDM_rate_2011,OFDM_PN_LO_ref_TWC_2014}. Given \eqref{eq_y_all}, \eqref{eq_rec_baseband_eve_y_re} and \eqref{eq_comm_final}, the SNRs are defined as  $\snra = |\alphaa|^2/\sigma^2$, $\snre=|\alphasetilde|^2/\sigma^2$ and $\snrc =|h_n|^2/\sigma^2 \, \forall n$, where a \ac{LOS}-only frequency-flat Alice-UE channel is assumed\footnote{Although adopted for simplicity, this channel model is not essential to the qualitative interpretation of the results. More general channel models would mainly affect the quantitative values of the achievable rate without changing the observed trends and trade-offs.}. In addition, target range $R$ refers to the monostatic range with respect to Alice as bistatic sensing performance is independent of $R$ for a fixed target SNR (see \eqref{eq_bxi_stat} and \eqref{eq_pn_stat_eve_dynamic}). Moreover, the \textbf{PN-free} bounds in Sec.~\ref{sec_crb_alice} and Sec.~\ref{sec_crb_eve} serve as an ideal baseline against the \textbf{PN-ignorant} bounds in Sec.~\ref{sec_mcrb_alice} and Sec.~\ref{sec_mcrb_eve}, which are computed as the \ac{LB} (the sum of \ac{MCRB} and the bias term \cite{Fortunati2017,PN_OFDM_ISAC_RadarConf_2023}) and averaged over $100$ PN realizations. The \ac{SPG} is defined as
\begin{align}\label{eq_SPG}
\mathrm{SPG}^{\rm dB}
=10\log_{10}\!\left(\frac{\mathrm{LB}_{\rm E}}
{\mathrm{LB}_{\rm A}}\right) \,,
\end{align}
where $\mathrm{LB}_{\rm E}$ and $\mathrm{LB}_{\rm A}$ denote the PN-ignorant \ac{LB} on range estimation at Eve and Alice, respectively\footnote{For consistency in terminology, we report both the PN-free and PN-ignorant bounds in terms of range \ac{RMSE}.}.

%%%%%%%%%%%%%%%%%%%%%%%%%%%%%%%%%
\begin{table}[t]\footnotesize
\caption{Simulation parameters}
\centering
\begin{tabular}{ll}
\hline
Carrier frequency $\fc$ & $28$ GHz\\
Subcarriers / spacing & $N=256$, $\deltaf=120$ kHz\\
Useful / CP duration & $T=8.33~\mu$s, $\Tcp=0.58~\mu$s\\
%Reference target range & $R=10$ m\\
%Pilot period & $\Kpil=8$\\
Oscillator quality & $\fdb = \fdbe=\fdbc=100$ Hz\\
\hline
\end{tabular}
\label{tab_parameters}
\end{table}
%%%%%%%%%%%%%%%%%%%%%%%%%%%%%%%%%%%%%

\subsection{Sensing Performance vs. SNR}\label{sec_result_snr}
Fig.~\ref{fig_results_sensing_snr} shows the ranging performance at Alice and Eve with respect to $\snr = \snra = \snre$ under PN-free and PN-ignorant processing. We observe that ignoring PN in sensing processing leads to performance saturation at high SNRs since PN becomes dominant over \ac{AWGN} with increasing SNR. Eve reaches this performance plateau at a lower SNR than Alice because Eve suffers from much more severe PN, as illustrated in Fig.~\ref{fig_pn_stats_contrast}, due to the use of an independent \ac{LO}. In addition, the ranging performance of PN-ignorant Alice degrades with increasing target range because of the range correlation effect, in alignment with \eqref{eq_acc_PN}.
%and Fig.~\ref{fig_pn_stats_contrast}.      

%%%%%%%%%%%%%%%%%%%%%%%%%%%%%%%%%%%%%%%%%%
% target RCS
\begin{figure}[t]
	\centering
    %\vspace{-0.1in}
	\includegraphics[width=1\linewidth]{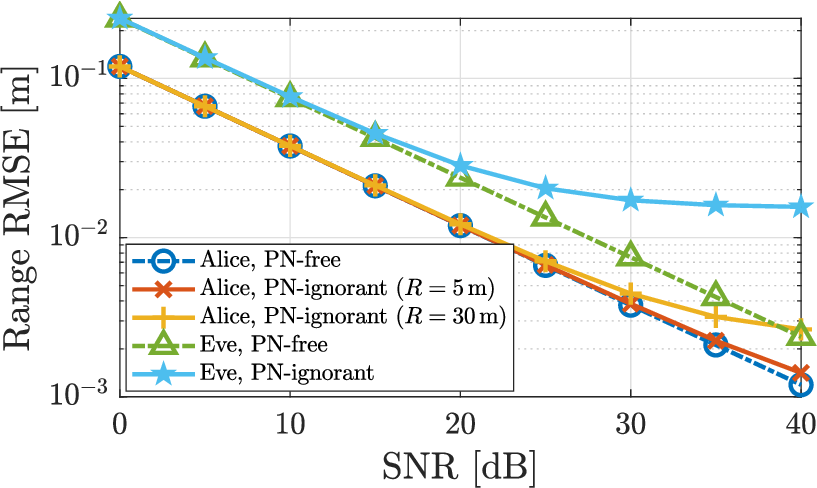}
	\vspace{-0.2in}
	\caption{Range \ac{RMSE} with respect to SNR with $\fdb = \fdbe = 100 \, \rm{Hz}$.} 
	\label{fig_results_sensing_snr}
	%\vspace{-0.1in}
\end{figure}
%%%%%%%%%%%%%%%%%%%%%%%%%%%%%%%%%%%%%%%%%%

\subsection{Sensing Performance vs. Oscillator Quality}\label{sec_result_lo}
To evaluate ranging accuracy under different levels of oscillator quality \cite{PN_OFDM_rate_2011,PN_OFDM_PLL_TCOM_2007 ,Demir_PN_2006}, we plot the theoretical bounds on range estimation with respect to the $3$-dB bandwidth of Alice's \ac{LO}, $\fdb$, while fixing that of Eve's to $\fdbe = 100 \, \rm{Hz}$, as shown in Fig.~\ref{fig_results_sensing_3db} with $\snra = \snre = 10 \, \rm{dB}$. In compliance with the \ac{LO} asymmetry depicted in Fig.~\ref{fig_pn_stats_contrast}, the PN-ignorant bound at Eve rises much faster than that at Alice. The TX-side PN at Alice contributes a full dynamic Brownian walk at Eve, but only a small self-referenced increment at Alice, as discussed in Sec.~\ref{sec_eve_pn_stat}. This leads to widening \ac{SPG}, as defined in \eqref{eq_SPG}, with increasing $\fdb$, providing a powerful privacy-enhancing mechanism in ISAC systems. Moreover, we observe smaller degradation in Alice's performance for closer targets as a result of delay-dependent PN statistics at Alice.

%%%%%%%%%%%%%%%%%%%%%%%%%%%%%%%%%%%%%%%%
\begin{figure}[t]
	\centering
    %\vspace{-0.1in}
	\includegraphics[width=1\linewidth]{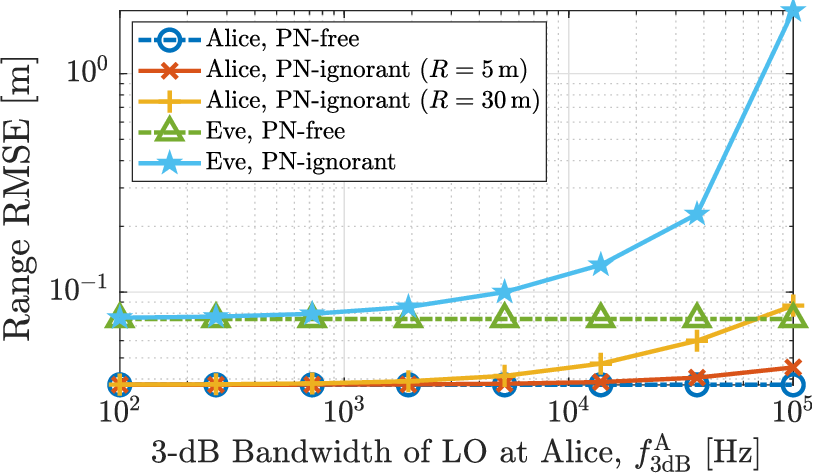}
	\vspace{-0.2in}
	\caption{Range \ac{RMSE} with respect to $3$-dB bandwidth of Alice's \ac{LO} while fixing that of Eve's to $\fdbe = 100 \, \rm{Hz}$ at $\snra = \snre = 10 \, \rm{dB}$.} 
	\label{fig_results_sensing_3db}
	%\vspace{-0.1in}
\end{figure}
%%%%%%%%%%%%%%%%%%%%%%%%%%%%%%%%%%%%%%%%

\subsection{Sensing Privacy Gap vs. Communication Rate}
In this part, we investigate the three-way trade-off among the sensing performance of Alice, the sensing performance of Eve and the communication performance of the Alice–UE link.  Fig.~\ref{fig_results_tradeoff} depicts the trade-off curves between the \ac{SPG} in \eqref{eq_SPG} and the achievable rate in \eqref{eq_cap}, obtained by sweeping $\fdb$ over $[100 \, \rm{Hz}, 100 \, \rm{kHz}]$. As $\fdb$ increases, the oscillator quality at Alice deteriorates, leading to a reduction in achievable rate while simultaneously enlarging the \ac{SPG}, in agreement with Fig.~\ref{fig_results_sensing_3db}. This reveals a fundamental trade-off in sensing-secure ISAC systems. Namely, enhancing sensing privacy comes at the expense of communication performance, with the oscillator quality serving as the key knob governing this balance. Consistent with the findings in Sec.~\ref{sec_result_snr} and Sec.~\ref{sec_result_lo}, more favorable trade-offs are observed for closer targets due to the range correlation effect. Moreover, at low communication SNRs, substantial gains in \ac{SPG} can be achieved with only a moderate reduction in achievable rate since the communication link is already operating in a noise-limited regime where the rate is mainly constrained by \ac{AWGN} rather than \ac{PN}.

%%%%%%%%%%%%%%%%%%%%%%%%%%%%%%%%%%%%%%%%
\begin{figure}[t]
	\centering
    %\vspace{-0.1in}
	\includegraphics[width=1\linewidth]{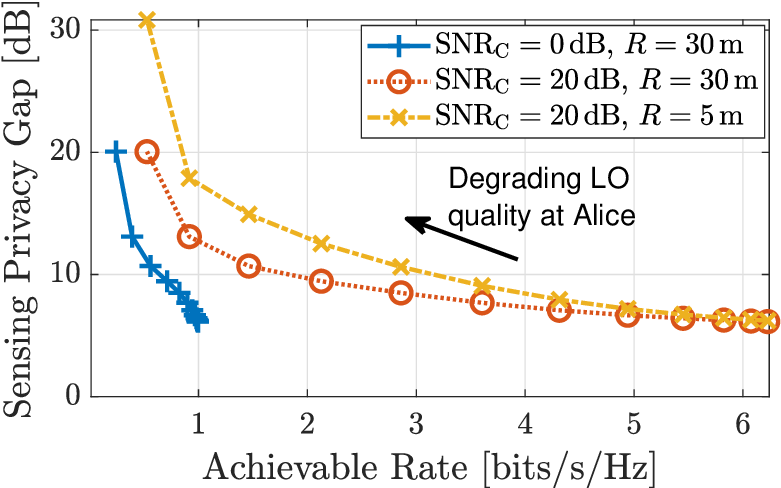}
	\vspace{-0.2in}
	\caption{Trade-offs between sensing privacy gap in \eqref{eq_SPG} and communication rate in \eqref{eq_cap} under various communication SNRs $\snrc$ and target ranges $R$ while fixing sensing SNRs to $\snra = \snre = 10 \, \rm{dB}$. The curves are obtained by sweeping $\fdb$ over $[100 \, \rm{Hz}, 100 \, \rm{kHz}]$.} 
	\label{fig_results_tradeoff}
	%\vspace{-0.1in}
\end{figure}
%%%%%%%%%%%%%%%%%%%%%%%%%%%%%%%%%%%%%%%%

%%%%%%%%%%%%%%%%%%%%%%%%%%%%%%%%%%%%%
\section{Conclusion}\label{sec_conclusion}
We have studied the role of \ac{PN} in sensing-secure OFDM ISAC systems with monostatic sensing at Alice, passive bistatic sensing at Eve and communication to a UE. The analysis shows that shared- and independent-\ac{LO} architectures induce fundamentally different PN statistics, enabling \ac{LO} asymmetry to enlarge the sensing privacy gap in favor of Alice. Simulations further reveal the resulting three-way trade-off among Alice sensing, Eve sensing and UE rate. Future work will consider PN-aware receivers, multi-target scenarios and range-Doppler processing.

\bibliographystyle{IEEEtran}
\bibliography{IEEEabrv,Sub/isac}

\end{document}